\documentclass{amsart}

\usepackage{bm}
\usepackage{amsmath,amsthm,amsmath}
\usepackage{algorithm, pseudocode}
\usepackage{mathrsfs}
\usepackage{enumitem}
\usepackage{fullpage}
\usepackage{xcolor}
\usepackage[bookmarksopen=false,pdftex=true,breaklinks=true,%
      plainpages=false,%
      hyperindex=true,pdfstartview=FitH,colorlinks=true,%
      pdfpagelabels=true,colorlinks=true,linkcolor=blue,%
      citecolor=red,urlcolor=blue,hypertexnames=false%
      ]%
   {hyperref}

\usepackage{multirow}

\def\headline#1{\hbox to \hsize{\hrulefill\quad\lower.3em\hbox{#1}\quad\hrulefill}}
\def\hline#1{\hbox to \hsize{\line(1,0){10}\quad \lower.3em\hbox{$#1$}\quad \line(1,0){10}}}

\definecolor{ao(english)}{rgb}{0.0, 0.5, 0.0}

\date{}

\newcommand{\X}{\mathbf{X}}
\newcommand{\Ww}{\mathcal{W}}
\newcommand{\Ss}{\mathcal{S}}
\newcommand{\Weyl}{\mathbf{W}}
\newcommand{\x}{\mathbf{x}}
\newcommand{\y}{\mathbf{y}}
\newcommand{\z}{\mathbf{z}}
\newcommand{\HH}{\mathbf{H}}
\newcommand{\Coxeter}{\mathcal{C}}
\newcommand{\R}{\mathbb{R}}
\newcommand{\W}{\mathcal{W}}
\def\N{\ensuremath{\mathbb{N}}}
\def\Q{\ensuremath{\mathbb{Q}}}

\DeclareMathOperator{\Comp}{Comp}
\newcommand{\comp}{\mathscr{C}}
\DeclareMathOperator{\sign}{sign}
\DeclareMathOperator{\thom}{Thom}
\DeclareMathOperator{\der}{Der}
\def\f{\ensuremath{F}}
\def\scrQ{\ensuremath{\mathscr{Q}}}

\DeclareBoldMathCommand{\p}{p}
\DeclareMathOperator{\RM}{RM}

\def\softO{\ensuremath{{O}{\,\tilde{ }\,}}}

\newcommand{\abb}[5]{%
\setlength{\arraycolsep}{0.4ex}%
\begin{array}{rcccc}%
#1 &:\,& #2 & \,\,\longrightarrow\,\, & #3 \\[0.5ex]%
     & & #4 & \longmapsto & #5%
\end{array}%
}

\newtheorem{definition}{Definition}
\numberwithin{definition}{section}
\newtheorem{theorem}[definition]{Theorem}
\newtheorem{corollary}[definition]{Corollary}
\newtheorem{proposition}[definition]{Proposition}
\newtheorem{lemma}[definition]{Lemma}

\newtheorem{example}[definition]{Example}
\newtheorem{notation}{Notation}[section]

\newtheorem{remark}[definition]{Remark}

\DeclareMathOperator{\CompMax}{CompMax}

\title[Deciding connectivity in Symmetric Semi-Algebraic Sets]{Deciding Connectivity in Symmetric Semi-Algebraic Sets} 
\thanks{C.\ Riener and R.\ Schabert were supported by the Troms\o{} Research Foundation grant 17MATCR.  A part of this work  was completed when T. X. Vu was working at the Department of Mathematics and Statistics, UiT The Arctic University of Norway, Tromsø, Norway, and at the Institute for Algebra, Johannes Kepler University, Linz, Austria. During that time, she was partially supported by the ANR-FWF grant 10.55776/I6130.
}

\author{Cordian Riener}
\address{Department of Mathematics and Statistics, UiT - the Arctic University of Norway, 9037 Troms\o, Norway}
\email{cordian.riener@uit.no}
\author{Robin Schabert}
\address{Department of Mathematics and Statistics, UiT - the Arctic University of Norway, 9037 Troms\o, Norway}
\email{robin.schabert@uit.no}
\author{Thi Xuan Vu}
\address{Univ. Lille, CNRS, Centrale Lille, UMR 9189 CRIStAL, F-59000 Lille, France}
\email{thi-xuan.vu@univ-lille.fr}

\begin{document}
\begin{abstract}
A \textit{semi-algebraic set} is a subset of $\mathbb{R}^n$ defined by a finite collection of polynomial equations and inequalities. In this paper, we investigate the problem of determining whether two points in such a set belong to the same connected component. We focus on the case where the defining equations and inequalities are invariant under the natural action of the symmetric group and where each polynomial has degree at most \( d \), with \( d < n \) (where \( n \) denotes the number of variables). Exploiting this symmetry, we develop and analyze algorithms for two key tasks. First, we present an algorithm that determines whether the orbits of two given points are connected. Second, we provide an algorithm that decides connectivity between arbitrary points in the set. Both algorithms run in polynomial time with respect to \( n \).
\end{abstract}

\maketitle

\section{Introduction}

A semi-algebraic set is a subset of \( \mathbb{R}^n \) that can be described using finite unions and intersections of polynomial inequalities of the form  
\[
\{ x \in \mathbb{R}^n \mid f(x) \geq 0 \},
\]
where \( f \) is a real polynomial. These sets naturally arise in real algebraic geometry and have numerous applications in both theoretical mathematics and computational problems.

One of the fundamental questions in the study of semi-algebraic sets is connectivity: given two points \( x, y \in S \), is there a continuous path within \( S \) that connects them? This problem has been extensively studied in algorithmic geometry and is central to applications such as robot motion planning \cite{schwartz1983piano}. In motion planning, connectivity determines whether a robot can transition between two configurations without leaving the feasible space, a setting where semi-algebraic sets naturally describe collision-free regions. Recently, the significance of roadmap algorithms has extended beyond classical robotics. Capco et al. \cite{capco2020robots, capco2023positive} demonstrated that roadmap-based methods can be effectively applied to kinematic singularity analysis, emphasizing the broader scope of connectivity problems in computational geometry. Additionally, the increasing trend toward numerical implementations (see, e.g., \cite{iraji2014nuroa}) highlights the practical necessity of developing robust connectivity algorithms with efficient computational performance.

In this article, we address the problem of deciding connectivity in symmetric semi-algebraic sets, that is, sets that are invariant under the action of the symmetric group \( S_n \). Such sets appear naturally in many-body physics, algebraic statistics, and computational invariant theory, where symmetries reduce complexity but also introduce new algorithmic challenges. Our main contribution is the development of a polynomial-time algorithm for deciding connectivity within a given symmetric semi-algebraic set. The key insight enabling this approach is a dimension reduction technique, which exploits the structure of symmetric polynomials to simplify the connectivity problem.

The remainder of this paper is organized as follows. In Subsection \ref{subsec:result}, we provide an overview of previous results on connectivity in semi-algebraic sets. Section \ref{sec:prelim1} introduces key preliminaries, including symmetric polynomials, the Weyl chamber, and Vandermonde maps. Section \ref{sec:prelim2} presents essential algorithmic background on methods that will be referenced throughout the paper.

In Section \ref{sec:graph}, we construct a framework that relates the connectivity of a union to connectivity in a bipartite graph. This construction, which follows directly from the Mayer–Vietoris exact sequence, plays a crucial role in our algorithms. Section \ref{sec:algo} presents the first main result of this paper, establishing a connection between connectivity in symmetric semi-algebraic sets and \emph{connectivity in orbits} (see Definition \ref{def:inorbt}). The latter is a weaker notion, as orbit connectivity is a consequence of general connectivity. We study this concept in detail and provide an algorithm for deciding connectivity in orbits. 

Section \ref{sec:general} introduces the notion of mirrored spaces and establishes a result that complements orbit connectivity to recover full connectivity. We then use these insights to present our main algorithm. Finally, we conclude with a discussion of open questions and future directions in the last section.

\subsection{Prior Results}
\label{subsec:result}
The connectivity problem has received significant attention; however, to the best of our knowledge, no previous work has taken the symmetry into account.

\medskip
\paragraph*{\em Computing Roadmaps and Deciding Connectivity.}
Schwartz and Sharir \cite{schwartz1983piano}  pioneered a solution based on Collins' cylindrical algebraic decomposition method. Their approach has a computational complexity that is polynomial in the degree $d$ and the number of polynomials $s$, but doubly exponential in the number of variables $n$.

Canny \cite{Canny} introduced the concept of a roadmap for a semi-algebraic set.  Roadmaps provide as a tool for counting connected components and verifying whether two points belong to the same connected component. Later improvements \cite{canny1993computing} enhanced this algorithm, enabling the construction of a roadmap for a semi-algebraic set defined by polynomials, where the sign-invariant sets provide a stratification of \( \mathbb{R}^n \), with a complexity of \( s^n (\log s) d^{O(n^4)} \). For arbitrary semi-algebraic sets, Canny applied perturbations to the defining polynomials, allowing the algorithm to determine whether two points lie in the same semi-algebraically connected component while maintaining the same complexity. However, this approach does not provide an explicit path between the points. Furthermore, a Monte Carlo version of the algorithm achieves a complexity of \( s^n (\log s) d^{O(n^2)} \).

Grigor'ev and Vorobjov \cite{grigor1992counting, canny1992finding} proposed an algorithm with complexity \( (sd)^{n^{O(1)}} \) for counting the connected components of a semi-algebraic set. Additionally, Heintz, Roy, and Solerno \cite{heintz1994single}, as well as Gournay and Risler \cite{gournay1993construction}, developed algorithms capable of computing a roadmap for any semi-algebraic set with the same complexity. Unlike Canny's algorithm, these approaches do not separate the complexity into a combinatorial part (depending on \(s\)) and an algebraic part (depending on \(d\)). Since a semi-algebraic set can potentially have up to \( (sd)^n \) connected components, the combinatorial complexity of Canny's algorithm is nearly optimal.

In \cite{basu2000computing}, a deterministic algorithm is introduced to construct a roadmap for any semi-algebraic set contained within an algebraic set of dimension \(k\), achieving a complexity of \(s^{k+1}d^{O(k^2)}\). This algorithm is particularly significant in robot motion planning, where the robot's configuration space is often modeled as a lower-dimensional algebraic set embedded in a higher-dimensional Euclidean space. A key advantage of this approach is that its complexity depends on the dimension of the algebraic set rather than the ambient space. The algorithm also achieves nearly optimal combinatorial complexity by utilizing only a fixed number of infinitesimal quantities, which reduces the algebraic complexity to \(d^{O(k^2)}\). Furthermore, it not only identifies connectivity but also computes a semi-algebraic path between input points if they belong to the same connected component, thereby providing the full scope of the problem.

However, none of the aforementioned algorithms achieves a complexity lower than \( d^{O(n^2)} \), nor do they produce a roadmap with a degree below \( d^{O(n^2)} \). Safey El Din and Schost \cite{safey2011baby} proposed a probabilistic algorithm that extended Canny's original approach to compute a roadmap for a closed and bounded hypersurface, with a complexity of \( (nd)^{O(n^{1.5})} \). Later, in \cite{basu2014baby}, the same authors, along with Basu and Roy, introduced a deterministic algorithm for computing a roadmap of a general real algebraic set, achieving a complexity of \( d^{O(n^{1.5})} \).

Basu and Roy, in \cite{basu2014divide}, introduced a divide-and-conquer strategy that reduces the dimension by half at each recursive step, producing a recursion tree of depth \( O(\log(n)) \). Their deterministic algorithm computes a roadmap for a hypersurface in time polynomial in \( n^{n\log^3(n)}d^{n\log^2(n)} \), with an output size polynomial in \( n^{n\log^2(n)}d^{n\log(n)} \). Notably, the algorithm's complexity is not polynomial in its output size. Furthermore, despite not assuming smoothness on the hypersurface, the algorithm can address systems of equations by considering the sum of squares of the polynomials.

Following the divide-and-conquer strategy, and under smoothness and compactness assumptions, Safey El Din and Schost \cite{din2017nearly} proposed a probabilistic roadmap algorithm with both output degree and running time polynomial in \( (nd)^{n\log(k)} \), where \( k \) denotes the dimension of the considered algebraic set. Note that, the algorithm in \cite{din2017nearly} achieves a running time that is subquadratic in the output size, with explicitly provided complexity constants in the exponent.

Recently, in \cite{prebet2024computing}, Pr{\'e}bet, Safey El Din, and Schost established a new connectivity result that extends the one in \cite{din2017nearly} to the unbounded case. In a separate work \cite{prebet2024part2}, they introduced an algorithm for computing roadmaps with the same complexity as that in \cite{din2017nearly}.

\medskip
\paragraph*{\em Systems of symmetric polynomials.} This paper explores the connectivity problem in cases where the polynomials describing a semi-algebraic set are symmetric and have low degree in relation to the number of variables. Several studies have highlighted the potential of leveraging the action of the symmetric group in such situations.

Timofte \cite{timofte2003positivity} established that for symmetric polynomials of lower degree, positivity can be deduced from positivity on low-dimensional test sets, which consist of points with at most half-degree distinct coordinates. Building on this foundation, \cite{riener2012degree, riener2016symmetric}, and \cite{riener2013exploiting} demonstrated that this approach can be widely applied to algorithmic methods for polynomial optimization, resulting in polynomial-time complexity. A similar approach was used by Basu and Riener to compute the equivariant Betti numbers of symmetric semi-algebraic sets \cite{basu2018equivariant}, the Euler-Poincaré characteristic \cite{basu2017efficient}, and the first \( \ell \) Betti numbers \cite{basu2022vandermonde} of such sets, achieving polynomial-time computation for fixed degree. Our approach is based on ideas inherent in these works as well, namely a reduction of dimension in the case of a degree which is lower than the number of variables. 

In more recent work, Faug{\`e}re, Labahn, Safey El Din, Schost, and Vu \cite{faugere2020computing}, along with Labahn, Riener, Safey El Din, Schost, and Vu \cite{labahn2023faster}, demonstrated that even when the degree is not fixed in advance, significant improvements in complexity can be achieved for certain algorithmic problems in real algebraic geometry.

\subsection{Our main results } The main results presented in this paper are described in Section \ref{subsec:result} and \ref{sec:general}. On the one hand, we present the first polynomial-time algorithm for deciding connectivity between two $S_n$-orbits of points in a symmetric semi-algebraic set. This algorithm is described in Algorithm \ref{algorithm:special} and its complexity is bounded by $O(n^{d^2})$ arithmetic operations in \(\Q\).  Building on this, we provide, on the other hand, the first polynomial-time algorithm for determining connectivity between two individual points in such a set. This algorithm is described in Algorithm \ref{algorithm:general} and has complexity $O\left(s^d n^{d^2} d^{O(d^2)} \right)$ arithmetic operations in \(\Q\).  Both of the quantities describing the complexity of our algorithms  are polynomial for fixed $d$.
\subsection{Relation to a previous version} This paper is an extended version of our previous work \cite{riener2024connectivity}, which appeared in the proceedings of the ISSAC 2024 conference. In this version, we provide more detailed explanations, expand upon our previous results, and introduce new findings.  In the conference version, we focused exclusively on the case where the points \( x = (x_1, \dots, x_n) \) and \( y = (y_1, \dots, y_n) \) were in canonical form, meaning they satisfied the ordering constraints \( x_1 \leq x_2 \leq \dots \leq x_n \) and \( y_1 \leq y_2 \leq \dots \leq y_n \). In contrast, the current paper shows how to remove this restriction and provides an algorithm for  general real points \( x \) and \( y \), thereby broadening the applicability of our approach. In particular, Section \ref{sec:general} introduces new results, including an algorithm for determining the connectivity of arbitrary points in symmetric semi-algebraic sets. Additionally, Sections \ref{sec:prelim2}, \ref{sec:graph}, and \ref{sec:algo} extend and generalize certain notations and foundational results. However, the constraint on the polynomial degree remains: the degree of the input symmetric polynomials must not exceed the number of variables in order to obtain algorithmic improvements.

\section{Algorithmic preliminaries}\label{sec:prelim1}

This section introduces key algorithmic techniques for analyzing connectivity in semi-algebraic sets. We first discuss the encoding of real algebraic points using univariate representations and Thom encodings. We then provide an overview of roadmap algorithms, summarizing essential results on computing roadmaps and deciding connectivity, which form the basis for the algorithms developed later in the paper.

\subsection{Encoding of real algebraic points}
Our algorithms manipulate points obtained through subroutines, with
the outputs having coordinates represented as real algebraic
numbers. We encode such a point using univariate polynomials and a
Thom encoding. 

Let $\mathcal{F} \subseteq \R[\X]$ and $x \in \R^n$. A mapping $\zeta:
\mathcal{F} \rightarrow \{-1, 0, 1\}$ is called a {\em sign condition}
on $\mathcal{F}$. The {\em sign condition realized by} $\mathcal{F}$
on $x$ is defined as $\sign(\mathcal{F}, x): \mathcal{F} \rightarrow
\{-1, 0, 1\}$, where $f \mapsto \sign(f(x))$. We say $\mathcal{F}$
{\em realizes} $\zeta \in \{-1, 0, 1\}^{\mathcal{F}}$ if and only if
$\sign(\mathcal{F}, x) = \zeta$.
\begin{definition}
    A {\em real univariate representation} representing $x \in \R^n$
consists of: 
\begin{itemize}
\item  a zero-dimensional parametrization $$\scrQ = (q(T), q_0(T),
  q_1(T), \dots, q_n(T)),$$ where $q, q_0, 
  \dots, q_n$ lie in $\Q[T]$ with gcd($q, q_0$) = 1, and  
\item a Thom encoding $\zeta$ representing an element $\vartheta \in
  \R$ such that  
\[ 
q(\vartheta) = 0 \quad \text{and} \quad x =
\left(\frac{q_1(\vartheta)}{q_0(\vartheta)}, \dots,
  \frac{q_n(\vartheta)}{q_0(\vartheta)}\right) \in \R^n. 
\] 
\end{itemize}
\end{definition}

Let $\der(q) := \{q, q^{(1)}, q^{(2)}, \dots, q^{(\deg(q))}\}$ denote
a list of polynomials, where $q^{(i)}$, for $i > 0$, is the formal
$i$-th order derivative of $q$. A mapping $\zeta$ is called a {\em
  Thom encoding} of $\vartheta$ if $\der(q)$ realizes $\zeta$ on $x$
and $\zeta(q) = 0$. Note that distinct roots of $q$ in $\R$ correspond
to distinct Thom encodings \cite[Proposition~2.28]{BPR06}. We refer
readers to \cite[Chapters 2 and 12]{BPR06} for details about Thom
encodings and univariate representations. 

\begin{example}
A real univariate representation representing $x = \left( \frac{1}{2}, \, 3/2 - \frac{\sqrt{2}}{4} \right)$ is $(\mathscr{Q}, \zeta)$, where
\[
\mathscr{Q} = (q, q_0, q_1, q_2) = (T^2 - 2, 2T, T, 3T - 1)
\]
and $\zeta = (1, 1)$ (equivalently, $\zeta = (+, +) = (q' > 0 \wedge q'' > 0)$).

Indeed, let $\vartheta = \sqrt{2}$ be a root of $q$. Then $x = \left( \frac{q_1(\vartheta)}{q_0(\vartheta)}, \, \frac{q_2(\vartheta)}{q_0(\vartheta)} \right)$. Moreover, $q'(T) = 2T$ and $q''(T) = 2$, implying that $q'(\vartheta) > 0$ and $q''(\vartheta) > 0$. Thus, $\zeta = (+, +)$ serves as a Thom encoding of $\vartheta$.
\end{example}

Given a non-zero polynomial $q \in \mathbb{R}[T]$, we consider the routine
\[
\text{{\sf ThomEncoding}}(q),
\]
which returns the ordered list of Thom encodings of the roots of $q$ in $\mathbb{R}$. This routine can be executed using \cite[Algorithm 10.14]{BPR06} with a complexity of $O(\delta_q^4 \log(\delta_q))$, where $\delta_q = \deg(q)$, involving arithmetic operations in $\mathbb{Q}$.

Let $q \in \mathbb{R}[T]$ be a polynomial of degree $\delta_q$ and $p \in \mathbb{R}[T]$ be a polynomial of degree $\delta_p$. Additionally, let $\thom(q)$ denote the list of Thom encodings of the set of roots of $q$ in $\mathbb{R}$. We define the routine 
\[
\text{{\sf Sign\_ThomEncoding}}(q, p, \thom(q)),
\]
which, for each $\zeta \in \text{{Thom}}(q)$ specifying the root $\vartheta$ of $q$, returns the sign of $p(\vartheta)$. This routine can be implemented using \cite[Algorithm 10.15]{BPR06} with a complexity of 
\[
O\left( \delta_q^2 \left( \delta_q \log(\delta_q) + \delta_p \right) \right)
\]
arithmetic operations in $\mathbb{Q}$.

\subsection{Roadmaps and connectivity}
In general, a roadmap for a semi-algebraic set is a curve that has a non-empty and connected intersection with all of its connected components. 

Formally, let $S \subset \mathbb{R}^n$ be a semi-algebraic set. We denote by $\pi$ the projection onto the $X_1$-axis, and for $x \in \mathbb{R}$, we define
\[
S_x = \{ y \in \mathbb{R}^{n-1} : (x,y) \in S \}.
\]
\begin{definition}
A roadmap for $S$ is a semi-algebraic set $\RM(S) \subset S$ of dimension at most one, which satisfies the following conditions: 
\begin{itemize}
  \item[{\sf RM}$_1$.] For every semi-algebraic connected component $C$ of $S$, the intersection $C \cap \RM(S)$ is a semi-algebraic connected set. 
  \item[{\sf RM}$_2$.] For every $x \in \mathbb{R}$ and for every semi-algebraic connected component $C'$ of $S_x$, we have $C' \cap \RM(S) \neq \emptyset$.
\end{itemize}
\end{definition}

Let $\mathcal{M} \subset \mathbb{R}^n$ be a finite set of points. A roadmap for $(S, \mathcal{M})$ is a semi-algebraic set $\RM(S, \mathcal{M})$ such that $\RM(S, \mathcal{M})$ is a roadmap of $S$, and $\mathcal{M} \subset \RM(S, \mathcal{M})$. Roadmaps can be used to decide the connectivity of semi-algebraic sets. We summarize below the results we use in this paper.

\begin{theorem}[{\cite[Theorem 3]{basu2000computing}} and {\cite[Theorem 16.14]{BPR06}}]
\label{them:normal}
  Let $F$ be a sequence of polynomials in $\mathbb{Q}[\X]$ with the algebraic set $Z \subset \mathbb{R}^n$ defined by $F$ of dimension $k$. Consider $G = (g_1, \dots, g_s) \subset \mathbb{K}[\X]$. Let $S$ be a semi-algebraic subset of $Z$ defined by $G$. Let $d$ be a bound for the degrees of $F$ and $G$. Consider a finite set of points $\mathcal{M} \subset Z$, with cardinality $\delta$, described by real univariate representations of degree at most $d^{O(n)}$. Then the following holds: 
  \begin{itemize}
    \item There is an algorithm whose output is exactly one point in every semi-algebraically connected component of $S$. The complexity of the algorithm is $s^{k+1}d^{O(n^2)}$ arithmetic operations in $\mathbb{Q}$.
    \item There exists an algorithm to compute a roadmap for $(Z, \mathcal{M})$ using $$\delta^{O(1)} d^{O(n^{1.5})}$$ arithmetic operations in $\mathbb{Q}$.
  \end{itemize}
  As a consequence, there exists an algorithm for deciding whether two given points, described by real univariate representations of degree at most $d^{O(n)}$, belong to the same connected component of $Z$, using $$s^{k+1}d^{O(n^2)}$$ arithmetic operations in $\mathbb{Q}$.
\end{theorem}
In the following sections, we use the notation ${\sf Normal\_Point\_Connected}(S)$ and ${\sf Normal\_Connect}(S, x, y)$ to denote the algorithm mentioned in the Theorem above.
${\sf Normal\_Point\_Connected}$ computes exactly one point per connected component of a semi algebraic set $S$ described by polynomials in $n$ variables of degree at most $d$. This algorithm returns at most $O((sd)^n)$ many points and its complexity is bounded by $s^{n}d^{n^2}$ operations in \(\Q\). Given a semi-algebraic set $S$ and two real points $x$ and $y$, the algorithm ${\sf Normal\_Connect}(S, x, y)$ returns \textbf{true} if $x$ and $y$ belong to the same connected component of $S$; otherwise, it returns \textbf{false}. The complexity of this algorithm is also  bounded by $s^{n}d^{n^2}$ operations in \(\Q\).  

Note also that there are some other prior results concerning the problem of computing roadmaps computing one point per each connected component, and deciding the connectivity of two given points (see Subsection~\ref{subsec:result}), with better complexities.  For example, when the input polynomials define an algebraic set, the complexities of these two procedure are $(nd)^{n \log n}$ by using a deteministic algorithm given in \cite{basu2014divide}.  The procedure returns $(nd) O(n \log n)$ many points. However, the result stated in the theorem  above is sufficient for our main algorithms in Section~\ref{sec:algo} and Section 5 and in the context of any semi-algebraic set.

\section{Preliminaries on symmetric polynomials}\label{sec:prelim2}
Symmetric polynomials play a fundamental role in our study of connectivity in semi-algebraic sets. In this section, we provide essential definitions and results on symmetric polynomials, their invariance properties under the symmetric group, and their representations using power sums. We also introduce Weyl chambers and their associated combinatorial structures, which allow us to analyze connectivity in symmetric semi-algebraic sets. Finally, we discuss Vandermonde maps, which provide a useful tool for encoding and understanding these spaces.

\subsection{Symmetric polynomials}
Throughout the article, we fix integers $n, d \in \mathbb{N}$ with $d \leq n$ and denote by  
\(
\mathbb{R}[\mathbf{X}] := \mathbb{R}[X_1,\dots,X_n]
\)
the real polynomial ring in $n$ variables. The \emph{basic closed semi-algebraic set} given by polynomials $f_1, \dots, f_s$ in $\mathbb{R}[\mathbf{X}]$ is  
\[
\{x \in \mathbb{R}^n \mid f_i(x) \geq 0 \text{ for all } i = 1, \dots, s\}.
\]
A semi-algebraic set is a set generated by a finite sequence of union, intersection
and complement operations on basic semi-algebraic sets.

For $n \in \mathbb{N}_0$, we denote by $S_n$ the symmetric group on $n$ elements, which has order $n!$. The symmetric group $S_n$ acts naturally on $\mathbb{R}^n$ by permuting coordinates, and this action extends to $\mathbb{R}[\mathbf{X}]$ by permuting the variables. A semi algebraic set $S\subset\R^n$ is called \emph{symmetric} if for every $x\in S$ and $\sigma\in S_n$ we have $\sigma(x)\in S$. We say that a polynomial $f \in \mathbb{R}[\mathbf{X}]$ is symmetric if  
\[
f(\mathbf{X}) = f(\sigma^{-1}(\mathbf{X})) \quad \text{for all } \sigma \in S_n,
\]
and we denote the $\mathbb{R}$-algebra of symmetric polynomials by $\mathbb{R}[\mathbf{X}]^{S_n}$. 

\begin{remark}
Note that clearly every semi-algebraic set which is exclusively defined by symmetric polynomials, is symmetric. While a symmetric semi-algebraic set can sometimes be described using polynomials that are not themselves symmetric, results by Bröcker \cite{Br} show that every symmetric semi-algebraic set can always be described by symmetric polynomials. Throughout this article, we assume that any symmetric semi-algebraic set is given by a description involving only symmetric polynomials.
\end{remark}
\begin{definition}
For $1 \leq i \leq n$, we define the \emph{$i$-th power sum} as  
\[
p_i := X_1^i + X_2^i + \cdots + X_n^i.
\]
\end{definition}

The fundamental theorem of symmetric polynomials states that every symmetric polynomial  can be uniquely written in terms of the first $n$ power sums, i.e., for every symmetric polynomial $f$, there is a unique polynomial $g \in \mathbb{R}[Z_1,\ldots,Z_n]$, where \(Z_1, \dots, Z_n\) are new variables, such that
\[
f = g(p_1,\ldots,p_n).
\]
Furthermore, due to the uniqueness of this representation, it follows that for a symmetric polynomial of degree $d < n$, this representation cannot contain power sums of degree higher than $d$, i.e., we have: 

\begin{proposition}
Any symmetric polynomial $f \in \mathbb{R}[\mathbf{X}]^{S_n}$ of degree $d \leq n$ can be uniquely written as
\[
f = g(p_1,\dots,p_d),
\]
where $g$ is a polynomial in $\mathbb{R}[Z_1,\dots,Z_d]$.
\label{lm:deg_restrict}
\end{proposition}

Note that there is no intrinsic preference for using the power sum polynomials $p_i$. Indeed, any family of algebraically independent symmetric polynomials, for example, the elementary symmetric polynomials and the complete homogeneous symmetric polynomials, with degree sequence $\{1,\ldots, n\}$, yields the same result.

\subsection{Weyl chambers and restriction to faces}

We denote by $\mathcal{W}_c$ the cone defined by $X_1 \leq X_2 \leq \cdots \leq X_n$, known as the \emph{canonical Weyl chamber} of the action of $S_n$. It serves as a fundamental domain for this action. The walls of $\mathcal{W}_c$ contain points where the inequalities are nonstrict at some positions, i.e., points $x = (x_1, \ldots, x_n) \in \mathcal{W}_c$ for which $x_i = x_{i+1}$ for some $i \in \{1, \ldots, n\}$. More precisely:

\begin{definition}
\label{def:composition-order}
Let  $n\in \N$ be a fixed integer.
\begin{enumerate}
    \item  A sequence of positive integers $\lambda = (\lambda_1, \ldots, \lambda_\ell)$ with $|\lambda| := \sum_{i=1}^{\ell} \lambda_i = n$ is called a \emph{composition of $n$ into $\ell$ parts}. Here, $\ell$ is called the \emph{length} \(\ell(\lambda)\) of $\lambda$, and we denote by $\Comp(n)$ the set of compositions of $n$ and by $\Comp(n,\ell)$ the set of compositions of $n$ of length $\ell$.

\item For $n \in \mathbb{N}$ and $\lambda = (\lambda_1, \ldots, \lambda_\ell) \in \Comp(n)$, we denote by $\mathcal{W}_c^{\lambda}$ the subset of $\mathcal{W}_c$ defined by
\[
X_1 = \cdots = X_{\lambda_1} \leq X_{\lambda_1+1} = \cdots = X_{\lambda_1+\lambda_2} 
\leq \cdots \leq X_{\lambda_1+\cdots+\lambda_{\ell-1}+1} = \cdots = X_n.
\]

\item For every $\lambda \in \Comp(n, d)$, the set $\mathcal{W}_c^\lambda$ defines a $d$-dimensional face of the cone $\mathcal{W}_c$, and every face is obtained in this way from a composition and denote by $L_\lambda$ the linear span of $\mathcal{W}_c^\lambda$. Note that
\[
\dim L_\lambda = \dim \mathcal{W}_c^\lambda = \ell(\lambda).
\] 
\item In the case when $\ell=n-1$,there is a unique $i\in\{1,\ldots,n-1\}$ such that $\lambda_i>1$. With a slight abuse of notation, we denote the corresponding facet of the canonical  Weyl chamber,  $\mathcal{W}_c^\lambda$,  also as $\mathcal{W}_c^{(i,i+1)}$. 
\item Finally, for every $d$, we define the set of \emph{alternate odd} compositions $\CompMax(n,d)$ as
\[
\left\{\lambda = (\lambda_1, \ldots, \lambda_d) \in \Comp(n) \; \mid \;
\lambda_{2i+1} = 1 \text{ for all } 0 \leq i < d/2 \right\}.
\]
\end{enumerate}
\end{definition}
Let $n \in \mathbb{N}$, and let $\lambda, \mu \in \Comp(n)$. We write $\lambda \prec \mu$ if $\mathcal{W}_c^{\lambda} \supset \mathcal{W}_c^{\mu}$. Equivalently, $\lambda \prec \mu$ if $\mu$ can be obtained from $\lambda$ by replacing some of the commas in $\lambda$ with $+$ signs. It follows that $\prec$ defines a \emph{partial order} on $\Comp(n)$, making $\Comp(n)$ a \emph{poset}.  
For $\lambda, \mu \in \Comp(n)$, the \emph{smallest composition greater than both} $\lambda$ and $\mu$ (with respect to $\prec$) is called their \emph{join}. The face corresponding to the join of $\lambda$ and $\mu$ is given by  
\(
\mathcal{W}_c^\lambda \cap \mathcal{W}_c^\mu.
\)

For $x \in \W_c$, we denote by $\comp(x)$ the largest composition $\lambda$ of $n$, such that $x$ can be written as 
    \[ x = (\underbrace{z_1, \dots, z_1}_{\lambda_1\text{-times}},
    \underbrace{z_2, \dots, z_2}_{\lambda_2\text{-times}}, \dots,
    \underbrace{z_\ell, \dots, z_\ell}_{\lambda_\ell\text{-times}}), \] 
    where $z = (z_1, \dots, z_\ell) \in \R^\ell$, and $\ell$ is the length of
    $\lambda$. For a composition $\lambda = (\lambda_1, \dots, \lambda_\ell)
      \in \Comp(n)$ and $f \in \R[\X]$, we define the polynomial
   \[ f^{[\lambda]} := f(\underbrace{X_1, \dots,
      X_1}_{\lambda_1\text{-times}}, \underbrace{X_2, \dots,
      X_2}_{\lambda_2\text{-times}}, \dots, \underbrace{X_\ell, \dots,
      X_\ell}_{\lambda_\ell\text{-times}}). \]

Notice that $f^{[\lambda]}$ is a polynomial in $\ell$ variables and its image is exactly the image of $f$ when restricted to $\W_c^\lambda$. More generally, given a symmetric semi algebraic set $S\subset\R^n$ we denote by $S_c=S\cap Ww_c$ and given a partition $\lambda \in \Comp(n)$ we denote $S_c^{[\lambda]}$ the intersection of $S_c$ with $\Ww_c^{\lambda}$. If $S$ is described by symmetric polynomials, then $S_c^{\lambda}$ is just obtained by replacing every polynomial $f$ in the description by $f^{[\lambda]}$.

\begin{example}
Consider $n = 4$.   There exist eight compositions of $4$ listed as
\[\left\{(4), (3,1),(1, 3), (2,2), (2,1, 1), (1,2,1), (1,1,2),
(1,1,1,1)\right\}.\] 
Notably, for  $\mu$ in the set  $\{(3,1),
(1,3), (4)\}$, we have  $(1,2,1) \prec \mu$. Furthermore, the join of
$(2,1,1)$ and $(1,1,1,1)$ is $(2,1,1)$. 

Given $x = (-1, 5 , 5, -3)$, we determine  $\comp(x) =
(1,2,1)$. Lastly, considering the polynomial $f = x_3^3 + x_1x_2 -
x_4$ and the composition $\lambda = (1,2,1)$, it follows that
$f^{[\lambda]} = x_1x_2 + x_2^3 - x_3$.   
\end{example} 

Given a general $x\in\R^n$ we know that there is a unique element $x'\in \mathcal{W}_c$ in the  $S_n$-orbit of $x$. This $x'$ is obtained by ordering the coordinates of $x$ increasingly. This can be done algorithmically with the classical  sorting algorithm bubble sort. This algorithm also can be used to obtain the unique permutation $\sigma\in S_n$ transforming $x$ into  $x'$. Since $\sigma\in S_n$ we know that we can write $\sigma$ as a product of pairwaise transposition $s_i=(i,i+1)$ for $i\in\{1,\ldots,n\}$ and we will use the following variant of Bubble sort to obtain the minimal set of transpositions necessary to describe $\sigma.$
\begin{algorithm}\label{algo:bubble}
\caption{{\sf Minimal\_Adjacent\_Transpositions}}
\begin{itemize}
\item[{\bf Input:}] Vector \( x = (x_1, \dots, x_n) \).
\item[{\bf Output:}] $T$ a minimal sequence of adjacent transpositions sorting \( x \) and the sorted vector \( x^{'} \).
\end{itemize}
\vspace{-0.1cm}
\noindent\rule{16.5cm}{0.5pt}
\begin{enumerate}
    \item Initialize an empty list \( T \) of transpositions.
    \item {\bf for} \( i = 1 \) to \( n-1 \) {\bf do}
    \begin{enumerate}
        \item {\bf for} \( j = 1 \) to \( n - i \) {\bf do}
        \begin{enumerate}
            \item {\bf if} \( x_j > x_{j+1} \) {\bf then}
            \begin{enumerate}
                \item Swap \( x_j \) and \( x_{j+1} \).
                \item Append transposition \( (j,j+1) \) to \( T \).
            \end{enumerate}
        \end{enumerate}
    \end{enumerate}
    \item Return \( (T, x) \).
\end{enumerate}
\end{algorithm}
\begin{theorem}
The above Algorithm  correctly computes the minimal decomposition of $\sigma$ into adjacent transpositions. Its complexity is $O(n^2)$.
\end{theorem}

\begin{proof}
Correctness follows from the observation that each adjacent transposition reduces the permutation's inversion number by exactly one. Since sorting requires eliminating all inversions, this process necessarily yields a minimal decomposition into adjacent transpositions. The complexity result is immediate, as the procedure is equivalent to bubble sort, which is well-known to have a worst-case complexity of $O(n^2)$.
\end{proof}
\begin{remark}
Note that there are other sorting algorithms with a better complexity. However,  while Bubble sort is not the most efficient sorting algorithm in general, the task of our Algorithm \ref{algo:bubble} is to obtain the  minimal sequence of adjacent transpositions that give the sorted vector. As more efficient sorting algorithms use swaps of non-adjacent coordinates, the inherent complexity of the requirement—specifically, determining the minimal sequence of adjacent transpositions—necessitates an $O(n^2)$ complexity.
\end{remark}

\subsection{Vandermonde maps}

\begin{definition}
For $m=(m_1,\ldots,m_n)\in \mathbb{N}^n$, we define the weighted power sum $p_{j}^{(m)}:=\sum_{i=1}^n m_iX_i^j$. 
\end{definition}
\begin{remark} Clearly, the usual power sum polynomials correspond to trivial weights $m=(1,\ldots,1)$. Furthermore, given a composition $\lambda$ of $n$ we directly see that the restriction of the $i$-th power sum polynomial $p_i$ to the Wall $W_c^\lambda$,  $p_j^{[\lambda]}$, as defined above coincided with $p_j^{(\lambda)}$.
\end{remark}
\begin{definition}
For a fixed $d \in \{1, \ldots, n\}$ and $m\in\N^n$, the $d$ (weighted)-\emph{Vandermonde map}:
\[
\abb{\nu_{n,d,m}}{\mathbb{R}^n}{\mathbb{R}^n}{x}{(p_1^{(m)}(x), \dots, p_d^{(m)}(x))}.
\]
When $m = (1, \ldots, 1)$, we simplify denote the map as $\nu_{n,d}$. When $d = n$, and $m=(1,\ldots,n)$ we simply refer to it as the Vandermonde map. 

For $a = (a_1, \dots, a_d) \in \mathbb{R}^d$, the fiber of the $d$-Vandermonde map, i.e., the set
\[
V(a) := \{ x \in \mathbb{R}^n \mid p_1^{(m)}(x) = a_1, \dots, p_d^{(m)}(x) = a_d \},
\]
is called the $m$ - weighted \emph{Vandermonde variety} with respect to $a$. Note that when $n = d$, and $m=(1,\ldots,1)$ every non-empty Vandermonde variety is precisely the orbit of a single point $x \in \mathbb{R}^n$.
\end{definition}
Arnold, Givental, and Kostov emphasized the significance of the Vandermonde map in their research on hyperbolic polynomials, and this topic has been explored further by \cite{Blekherman}. The required properties for our studies are summarized as follows.

\begin{theorem}[\cite{Blekherman, arnold1986hyperbolic,givental1987moments,kostov1989geometric}]\label{thm:arnold}
The Vandermonde map ($d=n,m=(1,\ldots,1)$) gives rise to a homeomorphism between a Weyl chamber  $\mathcal{W}$ and  its image. Moreover, for any composition $\lambda$, the weighted Vandermonde map $\nu_{n,\ell(\lambda),\lambda}$ provides a homeomorphism between $\mathcal{W}_c^{\lambda}$ and its image. Furthermore, for every $a\in \R^d$ and  $m\in\mathbb{Z}^d$ the $m$- weighted  Vandermonde variety $V_m(a)$  is either contractible or empty when intersected to the Weyl chamber $\mathcal{W}_c$. 
\end{theorem}

\section{Graphs and connectivity of unions}
\label{sec:graph}
A central part of our algorithm will require that we algorithmically understand the connectivity of unions of semi-algebraic sets. Our main tool for studying connectivity in a union is based on relating it to graphs, for which efficient algorithms for connectivity are available. First, recall the following definitions.  

\begin{definition}
Let $G=(V,E)$ be a graph. Given two vertices $u,v\in V$ we say that they are \emph{connected} if there exists a sequence of vertices \(v_1, v_2, \ldots, v_k\) where \(v_1 = u\), \(v_k = v\), and \((v_i, v_{i+1}) \in E\) for all \(1 \leq i < k\). Moreover, a \emph{connected component} of  \(G\) is a maximal set of vertices \(C \subseteq V(G)\) such that for every pair of vertices \(v, u\) in \(C\), there is a path in \(G\) between \(v\) and \(u\). 
\end{definition}

Our primary motivation for this concept arises from the following scenario. Consider a collection of semi-algebraic sets  $S_1,\ldots,S_k\subset\R^n$. Our goal is to analyze the connectivity properties of their union $\bigcup_{i=1}^k S_i$. To facilitate this analysis, we introduce a bipartite graph constructed as follows.

\begin{algorithm}
\caption{Construct Bipartite Graph from Semi-Algebraic Sets}\label{alog:G}
\begin{itemize}
    \item[{\bf Input:}] a collection of semi-algebraic sets $S_1, \ldots, S_k$
    \item[{\bf Output:}] a bipartite graph $G(A \cup B, E)$ where $A$ and $B$ are vertex sets representing connected components of the sets, and $E$ is a set of edges indicating connectivity between these components
\end{itemize}
\vspace{-0.1cm}
 \noindent\rule{16.5cm}{0.5pt}
 
\begin{enumerate}
    \item initialize $A, B, E$ as empty sets
    \item {\bf for} {each pair of indices $i, j$ with $1 \leq i, j \leq k$ and $i \neq j$} {\bf do}
    \begin{enumerate}
        \item     
        compute one point per connected component of  $S_{(i,j)} = S_i \setminus S_j$ and $S_{(j,i)} = S_j \setminus S_i$. Add the points together with the ordered pairs ${(i,j)}$ and ${(j,i})$ into $A$
        \item compute one point per connected component of $S^{i,j} = S_i \cap S_j$ and add each point to $B$ together with the tuple $(i,j)$
    \end{enumerate}
    \item {\bf for} each $u \in A$ {\bf do}
    \begin{enumerate}
        \item {\bf for} each $w \in B$ {\bf do}
        \begin{enumerate}
            \item  determine the semi-algebraic set $S_i$ to which both $u$ and $w$ belong by inspecting the associated ordered pair {$(i,j)$} and the associated tuple $\{u,w\}$
            \item {\bf if} {$u,w$ are in the same set $S_i$} {\bf then}
            \begin{enumerate}
                \item {\bf if} {\sc connect}({$u, w, S_i$}) {\bf then} add edge $(u, w)$ to $E$
            \end{enumerate}
        \end{enumerate}
    \end{enumerate}
    \item \textbf{return} $G(A \cup B, E)$
\end{enumerate}
\label{alg:ConstructGraph}
\end{algorithm}

This information on the pairs of connected components and their intersection contains enough information to deduce the connectivity of the union. Indeed, this can be seen directly via the usual Mayer-Vietoris spectral sequence, whose $(0.0)$-term of the spectral sequence stabilizes at the second page. From this sequence it follows that the $0$-cohomology of the union $\bigcup_i S_i$ is isomorphic to the kernel of
\[
\bigoplus_i \mathbb{H}^0(S_i) \rightarrow \bigoplus_{i,j} \mathbb{H}^0(S_{ij}),
\]
where the map is the usual generalized restriction map. This observation together with the notion of connectivity in graphs directly yields the following result.

\begin{theorem}\label{thm:graph}

Let \( S_1,\dots,S_k \subset \mathbb{R}^n \) be finite semi-algebraic sets, each of dimension \( n \), defined by polynomials of degree at most \( d \). The connected components of the union \( S = \bigcup_{i=1}^k S_i \) are in one-to-one correspondence with the connected components of the bipartite graph \( G \) constructed by Algorithm~\ref{alg:ConstructGraph}.
\end{theorem}
We give a short independent proof for the convenience of the reader. 
\begin{proof}
The vertex set $A$ represents points from the differences $S_i \setminus S_j$, $B$ represents points from the intersections $S_i \cap S_j$, and $E$ represents the connectivity between these points in $A$ and $B$. To show that we have an injective pairing: assume two distinct connected components in $S$, say $C_1$ and $C_2$, map to the same connected component in $G$. By construction, vertices in $A$ and $B$ correspond to unique connected components of $S_i \setminus S_j$ and $S_i \cap S_j$, respectively. Since $C_1$ and $C_2$ are distinct in $S$, they must yield distinct sets of points in $A$ and $B$, contradicting the assumption that they map to the same component in $G$. Hence, the mapping is injective.  To show the surjectivity of the pairing: consider a connected component in $G$, represented by a subset of vertices in $A \cup B$ and edges in $E$. By construction, each vertex in $A$ and $B$ corresponds to a connected component in the differences $S_i \setminus S_j$ or intersections $S_i \cap S_j$ of the semi-algebraic sets. Since edges in $E$ represent connectivity between these components, the corresponding points in $S$ form a connected subset. This subset corresponds to a connected component in $S$, ensuring that every component in $G$ maps back to a component in $S$, establishing surjectivity.

\end{proof}

\begin{theorem}[{\cite[Chapter 8]{jungnickel2005graphs}}]
Let \(G = (V, E)\) be an undirected graph. There exists an algorithm that computes all connected components of \(G\) with complexity \(O(|V| + |E|)\). Moreover, the  space complexity is \(O(|V| + |E|)\).
\end{theorem}
Using this graph algorithm in combination with Theorem \ref{thm:graph} we thus have a convenient  way of computing the connected components of a union.


\section{Connectivity in the orbit space}
\label{sec:algo}
In the presence of group invariance, it is natural to reformulate algorithms in a symmetry-reduced setup. For example, one is no longer interested in all solutions to a system of equations but only in all solutions up to symmetry. This corresponds to reducing the set of solutions to the orbit space, by which we mean the following 

\begin{definition}[Orbit Space]
Let \( G \) be a finite group acting on \(\mathbb{R}^n\). The \emph{orbit space}, denoted by \(\mathbb{R}^n/G\), is defined as the quotient space under the equivalence relation identifying points belonging to the same \(G\)-orbit. Formally,
\[
\mathbb{R}^n/G := \{G x : x \in \mathbb{R}^n\}, \quad\text{where}\quad Gx = \{g(x)\mid g \in G\}.
\]
\end{definition}
In particular, the orbit space can be identified with a fundamental domain and thus in the case of a group \(G\) generated by reflections , with  a \emph{Weyl chamber}. The notation of the orbit space motivates the following definition.
\begin{definition}[Connectivity up to Orbits]\label{def:inorbt}
Let \( G \) be a finite group acting on \(\mathbb{R}^n\), and let \( S \subseteq \mathbb{R}^n \) be a \( G \)-invariant semi-algebraic set. Two points \( x,y \in S \) are said to be \emph{connected up to orbits} (or \emph{orbit-connected}) if their orbits \( Gx \) and \( Gy \) lie in the same connected component of the orbit space \( S/G \). 
\end{definition}

To relate this definition to the connectivity of two points we remark the following connection between the two notions.

\begin{remark}\label{equivariant}
Let $S$ be an invariant semi-algebraic set $S$ and consider two  points $x,y\in S$. If \(x,y\) are connected by a continuous path \(\gamma : [0,1]\rightarrow S\), the path $\gamma$ itself descends naturally to a continuous path in the quotient space. Thus, the orbits \(Gx,Gy\) trivially lie within the same connected component in the orbit space \(S/G\). On the other hand,  two points \( x,y \in S \) may be connected up to orbits without necessarily being connected within the original set \(S\). For instance, distinct points in the same orbit are always orbit-connected but need not lie in the same connected component of \(S\).
\end{remark}
So while orbit connectivity of two points is not equivalent to the connectivity of the two points, the failure of orbit connectivity can be used to decide that the two points are not connected.  For the setup of this article, we fix the canonical  Weyl chamber $\mathcal{W}_c$ as a model of the orbit space. For every point $x\in\R^n$ obtain a unique representative of $x$ in $\mathcal{W}_c$ by ordering the coordinates of $x$ increasingly. Thus for  two points $x$ and $y$ in $S$ orbit connectivity is equivalent to stating that their unique representatives $x'$ and $y'$ are connected in $S\cap \mathcal{W}_c$. In the following we will thus  assume that the two points \(x = (x_1, \dots, x_n)\) and \(y = (y_1, \dots, y_n)\) are in the same fundamental domain of the action of the symmetric group, by assuming that the coordinates of the two given points are sorted in non-decreasing order. That is throughout we will assume that , \(x_1 \le x_2 \le \cdots \le x_n\) and \(y_1 \le y_2 \le \cdots \le y_n\).

\subsection{Restriction to subspaces}
The key to our algorithmic improvements in the case of orbit connectivity lies in a reduction of dimension. This reduction of dimension  rests on the following result of Arnold.

\begin{theorem}[{\cite[Theorem 7]{arnold1986hyperbolic}}]  \label{thm:arnold} 
\label{itemlabel:thm:arnold:b}
Let $a\in\R^d$, and suppose that the $d$-Vandermonde variety $V(a)\subset \R^n$ is non-singular. Then a point $x \in V(a)\cap \W_c$ is a minimizer of $p_{d+1}$ if and only if $x \in \W^\lambda_c$ for some $\lambda \in \CompMax(n,d)$.
\end{theorem}
This result also motivates the following.
\begin{definition}\label{def:Skd}
Given $n,d \in \mathbb{N}$, 
we denote  
\begin{eqnarray*}
\W_{c}^d &=& \bigcup_{\lambda \in \CompMax(k,d)} \W_c^\lambda.
\end{eqnarray*} 
Furthermore, for  a semi-algebraic subset $S \subset \mathbb{R}^n$, we denote

\begin{equation*}
\label{eqn:def:Skd}
S_{n,d} = S \cap \W_{c}^d,
\end{equation*}
which we call the \emph{$d$-dimensional orbit boundary of $S$}.
Notice that if $d \geq n$, then $S_{n,d} = S \cap \W_{c}$.
\end{definition}

The first claim below arises directly from Theorem \ref{thm:arnold}. Similarly, the subsequent statement can be deduced with Morse theoretic arguments built on Theorem \ref{thm:arnold} (see \cite[Proposition 7]{basu2018equivariant} for further specifics).

\begin{theorem} \label{thm:con}
Let $S\subset\mathbb{R}^n$ be a semi-algebraic set defined by polynomials of degree $d$. Then every connected component of $S\cap \W_c$ intersects with $S_{n,d}$. In fact, $S_{n,d}$ is a retraction of $S\cap\W_c$.
\end{theorem}
With the following Lemma and the subsequent Algorithm, we algorithmically obtain the dimension reduction. 
\begin{lemma} \label{lm:min_connect}
Let $S \subset \mathbb{R}^n$ be a semi-algebraic set defined by symmetric
polynomials of degree at most $d$ with $d \leq n$ and let $x\in S$. Furthermore, let $a=(p_1(x),\dots,p_d(x))$ and assume that $x' \in \mathbb{R}^n$ is a minimizer of $p_{d+1}$ in $V(a)$.  Then, $x$ and $x'$ are connected. 
\end{lemma}

\begin{proof}
First, we observe that for every $x \in S$ and every $a \in \mathbb{R}^d$, if
$x \in V(a)$, then $V(a) \subset S$. Indeed, let $z \in V(a)$. We will show that $z \in S$. By the definition of $V(a)$, all the Newton sums $p_1, \dots, p_d$ are constant on $V(a)$, that is, $p_i(z) = p_i(x)$ for all $i = 1, \dots, d$. Moreover, since $S$ is defined by
symmetric polynomials $F$ of degree at most $d$, by
Proposition~\ref{lm:deg_restrict}, each of these polynomials can be
expressed in terms of the first $d$ Newton sums $p_1, \dots, p_d$. That is 
$
F = G(p_1, \dots, p_d) 
$ for some polynomials $G$, which implies that 
$$
F(x) = G(p_1(x), \dots, p_d(x)) = G(p_1(z), \dots, p_d(z)) = F(z).
$$ 
Since $x \in S$, so is $z$.

Furthermore, since $V(a)$ is contractible and therefore connected, each $z \in V(a)$ belongs to the same connected
component as $x$. In particular, $x'$, which is in $V(a)$, and $x$ are in the same connected component.   
\end{proof}

\begin{theorem} \label{thm:Wel}
Let $V(a)$ be a non-empty Vandermonde variety, and let $\W_c$ be the
canonical Weyl-chamber. Then there is a unique minimizer of $p_{d+1}$
on $V(a)\cap \W_c \neq \emptyset$ with a multiplicity composition that
is bigger or equal to a composition $\lambda$ of length $d$ and 
$\lambda_{d}=\lambda_{d-2}=\lambda_{d-4}=\dots=1$.  

Furthermore, there exists an algorithm ${\sf Min\_Canonical}(a, \W_c)$ that takes
$a \in \mathbb{R}^d$ and the canonical Weyl chamber $\W_c$ as input and returns a
{real univariate representation} representing this minimizer of 
$p_{d+1}$ on $V(a) \cap \W_c$ with a time complexity of 
\[
O\left( {n - \lceil d/2 \rceil
  -1 \choose \lfloor  {d/2} \rfloor  -1} \, d^{4d+1} \log(d)\right)  =
\softO\left({n^{d/2} d^d}\right) 
\] arithmetic operations in $\mathbb{Q}$. 
\end{theorem}

\begin{algorithm}
\caption{{\sf Min\_Canonical}$(a, \W_c)$}
\begin{itemize}
    \item[{\bf Input:}] a point $a \in \R^d$ and the canonical Weyl-chamber $\W_c$ 
    \item[{\bf Output:}] a real univariate representation representing of the  minimizer of $p_{d+1}$ on $V(a) \cap \W_c$
\end{itemize}
\vspace{-0.1cm}
\noindent\rule{16.5cm}{0.5pt}

\hspace{-9cm}{\bf for} {$\lambda \in \Comp(n, d)$ with $\lambda_d =
  \lambda_{d-2} = \cdots  =1$} {\bf do}
\begin{enumerate}
\item compute polynomials $p^{[\lambda]} = (p_1^{[\lambda]}-a_1, \dots,
  p_d^{[\lambda]}-a_d)$
\item find a zero-dimensional parametrization $\scrQ$ of
  $V(p^{[\lambda]})$ 
\[\scrQ = (q(T), q_0(T),  q_1(T), \dots, q_d(T)) \in \Q[T]^{d+2}\]
\item define $\scrQ' \in \Q[T]^{n+2}$ with $$\scrQ' = \big(q(T), q_0(T), \underbrace{q_1(T), \dots,
      q_1(T)}_{\lambda_1\text{-times}}, \dots, \underbrace{q_d(T), \dots,
      q_d(T)}_{\lambda_d\text{-times}}\big) $$
\item compute $\thom(q) = {\sf ThomEncoding}(q)$ 
\item find  $\Sigma_0 = {\sf Sign\_ThomEncoding}(q, q_0, \thom(q))$ 
\item {\bf for} $i=1, \dots, d-1$ {\bf do} find $$\Sigma_i = {\sf Sign\_ThomEncoding}(q, q_i-q_{i+1},
  \thom(q))$$ 
\item  {\bf for} $i=1,\dots, d$   and $(\zeta_\vartheta,   \sigma_i) \in \Sigma_i$ {\bf do}
\begin{enumerate}
\item {\bf if} $\sigma_0 > 0 $  and $\sigma_i \le  0$ for all $i = 1, \dots,
  d-1$ \newline or  $\sigma_0 < 0$ and   $\sigma_i \ge  0$ for all $i
  = 1, \dots,  d-1$   
\begin{enumerate}
    \item {\bf return} $(\scrQ', \zeta_\vartheta)$
\end{enumerate} 
\end{enumerate}
\end{enumerate}
\label{alg:MV}
\end{algorithm}

\begin{proof}
The first part of the theorem is obtained from the
  results in \cite{arnold1986hyperbolic, meguerditchian1992theorem}. For the second part, we consider the
  following polynomial optimization   $$({\bf P}) \, :   \, \min_{x \in
    V(a)} \, p_{d+1},$$ where $V(a)$ is the  Vandermonde variety   with respect
  to $a$ and $p_{d+1} = X_1^{d+1} + \cdots + X_n^{d+1}$ is the Newton power sum of degree $d+1$.  Since \(V(a)\) is nonsingular, the minimizers of \(p_{d+1}\) on \(V(a)\) are the points at which the Jacobian matrix is not full-rank.
 
  Since the Jacobian of the map $\big(p_1(\X)-a_1, \dots, 
p_d(\X)-a_d, p_{d+1}(\X)\big)$ is \[ 
J = \begin{pmatrix}
1 & 1 & \cdots & 1 \\
2X_1 & 2X_2 & \cdots & 2X_n \\
\vdots & & & \vdots \\
(d+1)X_1^{d} & (d+1)X_2^{d} & \cdots & (d+1)X_n^{d}
\end{pmatrix},
\]  any $(d+1)$-minor of $J$ has the form
\[
\alpha \cdot \prod_{\substack{i_1 \le i < j \le i_{d+1},\\ (i_1, \dots,
    i_{d+1}) \subset 
     \{1, \dots, n\} }}(X_{i} - X_{j}),
\] where $\alpha \in \R_{\ne 0}$.  Therefore, for any point $x\in
\R^n$ with more than $d$ distinct coordinates, there exists a
$(d+1)$-minor of $J$ such that  this minor does not vanish at
$x$. This implies that any optimizer has at most $d$ distinct
coordinates. 

Let $\lambda = (\lambda_1, \dots, \lambda_d)$  be
a composition of  $n$ of length $d$ with $\lambda_d =
\lambda_{d-2} = \cdots =1$. Then the optimizers of the problem
$(\bf P)$ of type $\gamma$, where $\gamma$ is a composition of $n$ with
$\lambda \preceq \gamma$, have the form    
\[
x = (\underbrace{x_1, \dots,
      x_1}_{\lambda_1\text{-times}}, \underbrace{x_2, \dots,
      x_2}_{\lambda_2\text{-times}}, \dots, \underbrace{x_d, \dots,
      x_d}_{\lambda_d\text{-times}}), 
\] where  $\bar x = (x_1, \dots, x_d)$ is a solution of   
\begin{equation} \label{eq:zero}
p_1^{[\lambda]} - a_1 = \cdots = p_{d}^{[\lambda]} - a_d = 0. 
\end{equation} 
Since the composition has length $d$, the system \eqref{eq:zero} is
zero-dimensional in $d$ variables, and its solution set is given by a
zero-dimensional parametrization. We then check for the existence of a
real point $x = (x_1, \dots, x_d)$ in this solution set such that
$x_{i} \le x_{i+1}$ for $i=1, \dots, d-1$. To do this, we follow the
following steps.  

Let $\scrQ = (q(T), q_0(T), q_1(T), \dots, q_d(T))$ be a
zero-dimensional parametrization for the solution set of
\eqref{eq:zero}, and $\thom(q)$ be the ordered list of Thom encodings
of the roots of $q$ in $\R$. Suppose $\zeta_\vartheta \in \thom(q)$ is
a Thom encoding of a real root $\vartheta$ of $q$. Then 
\[\bar x = \left(\frac{q_1(\vartheta)}{q_0(\vartheta)}, \dots,
  \frac{q_d(\vartheta)}{q_0(\vartheta)}\right) \in \R^d.\]
Since $\gcd(q, q_0) = 1$, $q_0$ does not vanish at $\vartheta$. If
$q_0(\vartheta) > 0$, then $x_1 \le x_2 \le \dots \le x_d$ if and only
if $q_i(\vartheta) - q_{i+1}(\vartheta) \le 0$ for all $i=1, \dots,
d-1$. Otherwise, when $q_0(\vartheta) < 0$, $x_1 \le x_2 \le \dots \le
x_d$ if and only if $q_i(\vartheta) - q_{i+1}(\vartheta) \ge 0$ for
all $i=1, \dots, d-1$. Thus, it is sufficient to test the signs of
polynomials $q_0(T)$ and $q_i(T) - q_{i+1}(T)$ at
$\vartheta$. Computing $\thom(q)$ can be done using the ${\sf
  ThomEncoding}$ procedure, while checking the signs of $q_0$ and $q_i 
- q_{i+1}$ is obtained using the ${\sf Sign\_ThomEncoding}$
subroutine. In summary, the ${\sf Min\_Canonical}$ algorithm is given in
Algorithm~\ref{alg:MV}. 

\smallskip
We finish the proof by analyzing the complexity of the {\sf Min\_Canonical} algorithm. First since $\lambda_d = \lambda_{d-2} =
\cdots = 1$, it is sufficient to consider compositions of $n - \lceil
d/2 \rceil$ of length $\lfloor d/2 \rfloor$, which is equal to ${n -
  \lceil d/2 \rceil -1 \choose \lfloor d/2\rfloor-1}$.

For a fixed composition $\lambda$ of length $d$, finding a
zero-dimensional parametrization for the solution set of
\eqref{eq:zero} requires $O(D^3 + dD^2)$ arithmetic operations in
$\Q$, where $D$ is the number of solutions of
\eqref{eq:zero}, which is also equal to the degree of $q(T)$. This can
be achieved using algorithms such as those presented in 
\cite{rouillier1999solving}. Since
the system \eqref{eq:zero} is zero-dimensional with $d$ polynomials
and $d$ variables, B\'ezout's Theorem gives a bound for $D \le
d^d$. Therefore, the total complexity of finding a zero-dimensional
parametrization for the solution set of \eqref{eq:zero} is $O(d^{3d})$
arithmetic operations in $\Q$. 

Finally, the subroutine ${\sf ThomEncoding}(q)$, which returns the
ordered list of Thom encodings of the roots of $q$ in $\R$, requires
$$O(D^4 \log(D)) = O(d^{4d} \log(d^d)) = O(d^{4d+1} \log(d)).$$ Since
the degrees of $q_0$ and $q_i$ (for $i=1, \dots, d$) are at most that
of $q$, the procedure ${\sf Sign\_ThomEncoding}$ to determine the sign
of each of these polynomials at the real roots of $q$ is
$$O(D^2(D\log(D) + D)) = O(d^{2d}(d^{d}\log(d^d) +d^d)) =
O(d^{3d+1}\log(d))$$ operations in $\Q$. 

Thus, the total complexity of the $\sf Min\_Canonical$ algorithm is
 \[
O\left( {n - \lceil d/2 \rceil
  -1 \choose \lfloor  {d/2} \rfloor  -1} \, d^{4d+1} \log(d)\right)  =
\softO\left({n^{d/2} d^d}\right) 
\] arithmetic operations in $\Q$. This finishes our proof. 
\end{proof}

\subsection{Algorithm for testing the connectivity within one Weyl chamber}
We are now able to present the first main result  of this article: an algorithm designed to determine the connectivity in orbits of two points, or equivalently for two points in $\R^n$ with  coordinates ordered increasingly.

Consider two points $x = (x_1, x_2, \dots, x_n)$ and $y = (y_1, y_2, \dots, y_n)$ in $\R^n$, with their coordinates arranged in ascending order,  $x_1 \leq x_2 \leq \dots \leq x_n$ and $y_1 \leq y_2 \leq \dots \leq y_n$. Prior to introducing the algorithm, we will define some notations.

\begin{notation}
Consider $1\leq d \leq n$ to be an integer, an let $z\in\W_c$. Define $a_1=p_1(z),\ldots, a_d=p_d(z)$. We will denote the Vandermonde variety $V(a_1,\ldots,a_d)$ by $V_{d,z}$. Additionally, let $z'$ denote the unique point within $V_{d,z}$ that is guaranteed by Theorem \ref{thm:Wel}.
\end{notation}

\begin{theorem}\label{thm:conection}
Let $S \subset \R^n$ be a semi-algebraic set defined by symmetric
polynomials of degree at most $d$ with $d \leq n$, and take $x, y \in S \cap \W_c$. 
Then, $x$ and $y$ are connected within $S \cap \W_c$ if and only if the corresponding points $x'$ and $y'$ are connected within $S_{n,d}$.
\end{theorem}

\begin{proof}
By Lemma \ref{lm:min_connect}, $x$ and $y$ are both connected to $x'$ and $y'$, respectively. Furthermore, since by Theorem \ref{thm:con} $S_{n,d}$ is a retraction of $S \cap \W_c$, we have that $x'$ and $y'$ are connected within $S \cap \W_c$ if and only if they are connected within $S_{n,d}$.
\end{proof}

We can now formulate an algorithm for checking equivariant connectivity of semi-algebraic sets defined by symmetric polynomials of degree at most $d$.

\begin{algorithm}
\caption{{\sf Connectivity\_Symmetric\_Canonical}($\f, (x, y)$)}\label{algorithm:special}
\begin{itemize}
    \item[{\bf Input:}] {$\f = (f_1,\dots,f_s)\in \R[\X]^{S_n}$ of degrees at most $d$
  defining a semi-algebraic set $S$; two points  $x$  and $y$ in $\W_c$}
    \item[{\bf Output:}]  \textbf{true} if $x$ and $y$ are in the same connected
 of $S$;  otherwise return {\bf false}
\end{itemize}
\vspace{-0.1cm}
 \noindent\rule{16.5cm}{0.5pt}
\begin{enumerate}
    \item compute a list $L$ of alternate odd compositions of $n$ into $d$ parts
    \item {\bf for} $\lambda$ in $L$ {\bf do} compute $\f^{[\lambda]}:=(f_1^{[\lambda]},\dots,f_s^{[\lambda]})$.
    \item compute bipartite graph $G=(A\cup B, E)$ for $\bigcup_{\lambda \in L} S(\f^{[\lambda]})$ \\
    \item compute connected components of $G$ \\ 
    \item compute $a = (\nu_{n,d}(x))$ and $b = (\nu_{n,d}(y))$
    \item compute  $x' = {\sf Min\_Canonical}(a, \W_c)$ and $y' = {\sf Min\_Canonical}(b,
      \W_c)$ 
    \item find the compositions $\gamma = \comp(x')$ of $x'$ and $\eta = \comp(y')$ of $y'$
    \item find $v_x,v_y \in A$ with 
    \begin{itemize}
        \item  ${\sf Normal\_Connect}(\f^{[\gamma]}, (v_x,x'))={\bf true}$ and 
        \item ${\sf Normal\_Connect}(\f^{[\eta]}, (v_y,y'))={\bf true}$.
    \end{itemize}
    \item {\bf if} $v_x$ and $v_y$ are in the same connect component of $G$ return {\bf true}, else return {\bf false}
\end{enumerate}
\end{algorithm}
\begin{theorem}
Given symmetric polynomials $f_1, \dots, f_s \in \R[\X]^{S_n}$ defining a semi-algebraic set $S$ and two points $x, y \in S \cap \W_c$, the Algorithm \ref{algorithm:special} correctly decides connectivity. Moreover, for fixed $d$ and $s$, its complexity can be estimated by $O(n^{d^2})$.
\end{theorem}

\begin{proof}
The correctness follows from Lemma \ref{lm:min_connect} in combination with Theorem \ref{thm:graph}. For the complexity of the algorithm, we observe that the main cost lies in building the graph $G$ with the help of Algorithm \ref{alog:G}. Notice that this algorithm uses a roadmap algorithm for the $d$ -dimensional semialgebraic sets defined by $\f^{[\lambda]}$. Therefore, following Theorem \ref{them:normal}, each call is made to compute the points per connected component or to decide the connectivity costs $s^{d+1} d^{O(d^2)}$, which are constant by assumption.

Moreover, the number of sets we consider
\[
\kappa(d) = O\left( {n - \lceil d/2 \rceil - 1 \choose \lfloor d/2 \rfloor - 1} \, d^{4d+1} \log(d) \right) = \softO\left( n^{d/2} d^d \right),
\]
which is the number of alternate odd compositions of $n$ with $d$ parts. Thus, the first loop has length $\binom{\kappa(d)}{2}$. Each step of the loop has a fixed complexity, which is at most $3s^{d+1} d^{O(d^2)}$. 

For the second loop, note that we can bound the number of connected components of each of the sets $S(\f^{[\lambda]})$ from above by $O((sd)^d)$. Therefore, the second loop has length at most $\kappa(d)^2 O((sd)^{2d})$, and each step has complexity at most $s^{d+1} d^{O(d^2)}$. Since $s$ and $d$ are fixed, we arrive at the announced complexity.
\end{proof}


\section{Connectivity in mirror spaces}
\label{sec:general}
In the prior section, we provided an effective method for determining connectivity within the canonical Weyl chamber. However, as highlighted  in Remark \ref{equivariant}, establishing orbit connectivity of two points does not in itself guarantee that the two points are connected. In order to extend our efficient algorithm to encompass the general situation of points situated in different Weyl chambers, we will use the notion of a \emph{mirror space}. This construction allows us to see how the topology of the orbit space, i.e.,  the topology of the intersection with the canonical Weyl chamber, connects to the topology of the entire set $S$. We will first give a short introduction to the necessary notions and results needed and deduce some results for our setting.  
\subsection{Mirror spaces}
The visual idea behind mirror spaces can be compared to a kaleidoscope: the walls of the Weyl chamber act as mirrors and these mirrors reflect whatever is inside the Weyl chamber to build up the entire symmetric set. Depending on how a symmetric set $S$ intersects with the Weyl chamber, one can describe the topology of $S$ solely by its orbit space. The notion of mirror space  is more generally defined in the setup of Coxeter groups. 
\begin{definition}
A Coxeter pair $(\Ww,\Ss)$, consists of a group $\Ww$ and a set of generators, $\Ss = \{ s_i \mid i \in I\}$, of $\Ww$ each of order $2$,
and numbers $(m_{i,j})_{i,j \in I}$ such that $(s_i s_j)^{m_{ij}} = e$.
\end{definition}
In the context of this paper, the example of a  Coxeter group which will be interesting for us is  the symmetric group ${S}_n$ 
considered as a Coxeter group with
the set of Coxeter generators $\Coxeter(n) = \{ s_i = (i,i+1) \mid 1 \leq i \leq n-1 \}$
(here $(i,i+1)$ denotes the permutation of $(1,\ldots,n)$ that exchanges $i$ and $i+1$
keeping all other elements fixed). Now following notion makes our above vague description of a mirrored space more mathematically precise. The reader is invited  to consult \cite{Davis-book} for more details.

\begin{definition}[Mirrored space]
\label{def:mirrored-space}
Given a Coxeter pair  $(\Ww,\Ss)$  (i.e. $\Ww$ is a Coxeter group and $\Ss$ a set of reflections generating $\Ww$)
a space $Z$ with a family of closed 
subspaces $(Z_s)_{s\in \Ss}$ 
is called a \emph{mirror structure} on $Z$ \cite[Chapter 5.1]{Davis-book},
and $Z$ along with the collection $(Z_s)_{s \in \Ss}$ is called a \emph{mirrored space}  over $\Ss$.
\end{definition}

Given a mirrored space  $Z, (Z_s)_{s \in \Ss}$ over $\Ss$,
the following so-called `Basic Construction'  (see in \cite[Chapter 5]{Davis-book})
gives  a space $\mathcal{U}(\Ww,Z)$ with a $\Ww$-action. This construction is defined as follows.

\begin{definition}[The Basic Construction \cite{Koszul, Tits, Vinberg, Davis}]
\label{def:U}
We define
\begin{equation}
\label{eqn:U}
\mathcal{U}(\Ww,Z) = \Ww \times Z/\sim
\end{equation}
where the topology on $\Ww \times Z$ is the product topology, with $\Ww$ given the discrete topology,
and the equivalence relation $\sim$ is defined by  
\[
(w_1,\x) \sim (w_2,\y) \Leftrightarrow \x= \y \mbox { and } w_1^{-1}w_2 \in \Ww^{\Ss(\x)},
\]    
with 
\[
\Ss(\x) = \{s \in \Ss \; \mid \; \x \in Z_s\},
\]  and $\Ww^{\Ss(\x)}$ the subgroup of $\Ww$ generated by $\Ss(\x)$. 

The group $\Ww$ acts on $\mathcal{U}(\Ww,Z)$ by $w_1 \cdot [(w_2,\z)] =  [(w_1w_2,\z)]$ (where
$[(w,\z)]$ denotes the equivalence class of $(w,\z) \in \Ww \times Z$ under the relation $\sim$).
\end{definition}
The action of $\Ww$ on $\mathcal{U}(\Ww,Z)$ now naturally gives rise also to an action on the cohomology groups $\HH_*(\mathcal{U}(\Ww,Z))$ which thus obtain the structure of a finite dimensional $\Ww$-module. Understanding the decomposition of this module in terms of other $\Ww$-modules is studied in \cite{Davis-book} in the case where $Z$ is a finite CW-complex. The following Theorem is a fundamental result of \cite{Davis}. We state the theorem in the special case where
$(\Ww,\Ss) = (\mathfrak{S}_k,\Coxeter(k))$ which is the case of interest to us. In this case, we obtain a decomposition of $\HH_*(\mathcal{U}(\Ww,Z))$ into so-called  \emph{Salomon modules}  $\Psi^{(k)}_{T}$ which are indexed by subsets of $S$ and defined in \cite{Solomon1968}. For our  considerations the detailed definition of these modules is not relevant - and the interested reader is invited to consult \cite{Davis-book,basu2022vandermonde,Solomon1968} for details on these modules - and the relation to irreducible $S_n$ representations. For our purposes it is sufficient to note that  $\Psi_{T}^{(k)}  \cong_{\mathfrak{S}_n}1_{\mathfrak{S}_n}$  if and only if $T=\emptyset$ (See \cite[Proposition 3]{basu2022vandermonde}).

\begin{theorem}\cite[Theorem 15.4.3]{Davis-book}
\label{thm:Davis0}
Let $(\Ww,\Ss) = (\mathfrak{S}_k,\Coxeter(k))$, and $Z,Z_s, s \in \Ss$ a semi-algebraic mirrored space over $\Ss$,  and  $Z,Z_s, s\in \Ss$ closed and bounded.
Then,
\[
\HH_*(\mathcal{U}(\Ww,Z)) \cong_{\mathfrak{S}_k}  \bigoplus_{T \subset \Ss} \HH_*(Z, Z^T) \otimes  \Psi^{(k)}_{T},
\]
where for each $T \subset \Ss$,
\[ 
Z^T = \bigcup_{s \in T} Z_s.
\]
\end{theorem}
For the above-mentioned result on Salomon modules we obtain the following characterization when a connected $Z$ leads to a connected  $\mathcal{U}(\Ww,Z)$.
\begin{theorem}
\label{thm:connected_mirrored_space}
With the notations as in the above Theorem. Let \( (\Ww,\Ss) = (\mathfrak{S}_k,\Coxeter(k)) \), and let \( Z \) be a connected, closed, and bounded semi-algebraic mirrored space over \( \Ss \). 

Then, the mirrored space \( \mathcal{U}(\Ww,Z) \) is connected if and only if:
\[
\HH_0(Z, Z^T) = 0 \quad \text{for all } T \neq \emptyset.
\]

Moreover, this condition holds if and only if \( Z \) and \( Z^T \) are path-connected for every \( T \neq \emptyset \).
\end{theorem}

\begin{proof}
From Theorem~\ref{thm:Davis0}, we have the homology decomposition
for the zero-th homology group \( \HH_0 \):
\[
\HH_0(\mathcal{U}(\Ww,Z)) \cong_{\mathfrak{S}_k}  \bigoplus_{T \subset \Ss} \HH_0(Z, Z^T) \otimes  \Psi^{(k)}_{T}.
\]
Clearly,  \( \mathcal{U}(\Ww,Z) \) is connected if and only if \( \HH_0(\mathcal{U}(\Ww,Z)) \) has rank 1. This happens if and only if all additional contributions from \( T \neq \emptyset \) vanish, i.e., 
\[
\HH_0(Z, Z^T) = 0 \quad \text{for all } T \neq \emptyset.
\]
By the definition of relative homology, \( \HH_0(Z, Z^T) \) is trivial if and only if every connected component of \( Z \) intersects \( Z^T \). Since \( Z \) is connected by assumption, this reduces to checking that \( Z^T \) is connected with \( Z \) for each \( T \neq \emptyset \). Thus,   the path-connectedness criterion gives that the relative homology \( \HH_0(Z, Z^T) = 0 \) if and only if \( Z \) and \( Z^T \) are path-connected for all \( T \neq \emptyset \). 
\end{proof}

\begin{corollary}
\label{cor:transposition_connectivity}
Let \( (\Ww,\Ss) = (\mathfrak{S}_k,\Coxeter(k)) \), and let \( Z \) be a connected, closed, and bounded semi-algebraic mirrored space over \( \Ss \). 

Let \( \sigma \in W \) be expressed as a product of transpositions:
\[
\sigma = t_1 t_2 \cdots t_k.
\]
Then, the points \( x \) and \( \sigma(x) \) are in the same connected component of the mirrored space \( \mathcal{U}(W,Z) \) if and only if:
\[
H_0(Z, Z^T) = 0 \quad \text{for all } \{t_1, \dots, t_k\} \subset T.
\]
\end{corollary}

\begin{proof}
By Theorem~\ref{thm:connected_mirrored_space}, the mirrored space \( \mathcal{U}(\Ww, Z) \) is connected if and only if:

\[
H_0(Z, Z^T) = 0 \quad \text{for all } T \neq \emptyset.
\]

For the points \( x \) and \( \sigma(x) \) to belong to the same connected component of \( \mathcal{U}(W, Z) \), there must exist a path in \( \mathcal{U}(W,Z) \) from \( x \) to \( \sigma(x) \). Since \( \sigma \) is expressed as a product of transpositions, this requires that we can connect \( x \) to \( t_1(x) \), then to \( t_2 t_1(x) \), and so on, until reaching \( \sigma(x) \).

By the previous theorem, the necessary and sufficient condition for connectivity is that the relative homology group \( H_0(Z, Z^T) \) vanishes for every set \( T \) containing the transpositions involved in \( \sigma \), i.e., for all \( \{t_1, \dots, t_k\} \subset T \).
Thus, \( x \) and \( \sigma(x) \) are in the same connected component of \( \mathcal{U}(W,Z) \) if and only if:
\[
H_0(Z, Z^T) = 0 \quad \text{for all } \{t_1, \dots, t_k\} \subset T
\] which finishes our proof.
\end{proof}

To connect these results to the setup of this paper, suppose that $S$ is a closed and bounded \emph{symmetric}  semi-algebraic subset of $\R^n$. Denoting by $S_c \subset \Ww_{c}$ the intersection of $S$ with $\Ww_c$ along with the  tuple of  closed semi-algebraic subsets $(S_{c,s} = S_c \cap \Ww_{c}^{s})_{s \in \Coxeter(n)}$ is a semi-algebraic mirrored  space over $\Coxeter(n)$ and the semi-algebraic set $\mathcal{U}(\mathfrak{S}_k,S_k)$ is semi-algebraically homeomorphic to 
$S$ \cite{basu2022vandermonde}. To conclude this excursion into mirrored spaces we now present the characterization of connectivity in the $S$ in terms of connectivity in $S_c$.

\begin{definition}
Let \( S \subseteq \mathbb{R}^n \) be a semi-algebraic set, and let \( S_c = S \cap \mathcal{W}_c \). Suppose \( S_c \) is connected. Define the subgroup \( G(S_c) \subseteq S_n \) by:
\[
  G(S_c) = \langle (i, i+1) \mid S_c \cap \Ww^{(i,i+1)} \text{ is connected} \, \rangle.
\]
This subgroup \( G(S_c) \) is the largest subgroup of \( S_n \) preserving the connectivity of \( S_c \).
\end{definition}

We apply this definition to obtain the following characterization:

\begin{theorem}\label{thm:connectivity}
Let \( S \subset \mathbb{R}^n \) be a closed, bounded, and symmetric semi-algebraic set. Consider \( x \in S_k \) and \( y \in S \). Denote by \( y' \in S_c \) the unique element in the orbit of \( y \). Then \( x \) and \( y \) belong to the same connected component of \( S \) if and only if:
\begin{enumerate}
    \item \( x \) and \( y' \) lie in the same connected component of \( S_c \).
    \item The permutation \( \sigma \) mapping \( y \) to \( y' \) belongs to \( G(\mathcal{C}) \), where \( \mathcal{C} \subseteq S_c \) is the connected component of \( S_c \) containing \( x \).
\end{enumerate}
\end{theorem}
\begin{proof}
Assume that $x$ and $y$ are connected by a path $\phi$. Any  path $\phi$ connecting $x$ and $y$ can be decomposed into a sequence of paths $\phi^\sigma$  where $\sigma\in S_n$ and  $\phi^\sigma=\phi\cap \sigma \Ww_c$. Thus each $\sigma^{-1}(\phi^\sigma)$ is a path in $\Ww_c$ and the union of these  connects $x$ to $y'$. Thus we have directly that the connection of $y$ and $x$ yields the connection of $x$ and $y'$. On the other hand, by Corollary~\ref{cor:transposition_connectivity}, the points \( x \) and \( \sigma(x) \) belong to the same connected component of the mirrored space \( \mathcal{U}(W,Z) \) if and only if:
\[
H_0(Z, Z^T) = 0 \quad \text{for all } \{t_1, \dots, t_k\} \subset T.
\]
Since \( S \) is semi-algebraically homeomorphic to \( \mathcal{U}(\mathfrak{S}_k, S_k) \), this connectivity condition translates into the semi-algebraic setting.
To show that \( x \) and \( y \) belong to the same connected component of \( S \), we proceed as follows: We have that  \( x \) and \( y' \) lie in the same connected component of \( S_c \). The permutation \( \sigma \) sending \( y \) to \( y' \) belongs to \( G(\mathcal{C}) \), where \( \mathcal{C} \) is the connected component of \( S_c \) containing \( x \). By construction of \( G(S_c) \), this subgroup is generated by transpositions preserving connectivity within \( S_c \). Thus, membership in \( G(S_c) \) ensures that connectivity is maintained under the action of \( \sigma \).
\end{proof}

\subsection{Algorithm for checking connectivity}
We are now in the position to give an algorithm that takes two points $x,y\in S$ and decides if they are connected in $S$. Without loss of generality, we can assume that $x$ has ordered coordinates, i.e., $x\in S_c$, since otherwise we can reiterate the same argumentation by replacing the canonical Weyl chamber $\Ww_c$ with the Weyl chamber that $x$ is contained in. Our approach now builds on Theorem \ref{thm:connectivity}. For this, we need to find the set of facets of $\Ww$ to  which $x$ can be connected in $S_c$.
As a observation  for this, we use the following reduction which is a consequence of applying  Arnold's result to the intersection of $S_c$ with the face $\Ww^{(i,i+1)}$.
\begin{definition}\label{def:res}
Let \( n, d \in \mathbb{N} \) and let \( \lambda = (\lambda_1, \lambda_2, \dots, \lambda_d) \) be a composition of \( n \) into \( d \) parts.
For a given index \( i \in \{1, \dots, n-1\} \), we define \( \textsc{Res}_i(\lambda) \), the composition of \( n \) that additionally incorporates the constraint \( x_i = x_{i+1} \), as follows. 
Let 
\[
\pi_k = 1 + \sum_{j=1}^{k-1} \lambda_j, \quad \text{for } k \geq 2, \quad \pi_1 = 1, 
\]
and let \( k \) and \( k' \) be the unique indices such that
\[
\pi_k \leq i < \pi_{k+1}, \quad \pi_{k'} \leq i+1 < \pi_{k'+1}.
\]
Then, 
\[
\textsc{Res}_i(\lambda)=
\begin{cases}
\lambda, & \text{if } k = k', \text{ (i.e., \( x_i = x_{i+1} \) is already implied)} \\
(\lambda_1, \dots, \lambda_{k-1}, \lambda_k + \lambda_{k'}, \lambda_{k'+1}, \dots, \lambda_d), & \text{otherwise}.
\end{cases}
\] 
\end{definition}
The above Definition also gives a direct algorithm to compute this new composition out of a given $\lambda$ in $O(d)$.
\begin{lemma}\label{lem:arnold2}
Let \( S \subset \mathbb{R}^n \) be a symmetric semi-algebraic set defined by symmetric polynomials of degree at most \( d \). Then, every connected component of \( S_c \cap W_{(i,i+1)} \) is connected to a point in \( S_c^{\textsc{Res}(\lambda)} \) for some \( \lambda \in \CompMax(n,d) \).
\end{lemma}

\begin{proof}
For any \( x \in S_c \cap W_{(i,i+1)} \), consider the composition \( m \in \mathbb{N}^{n-1} \) defined by
\[
m_j =
\begin{cases}
2, & \text{if } j = i, \\
1, & \text{otherwise}.
\end{cases}
\]
Define the associated symmetric polynomial invariants
\[
a_j = p_j^{(m)}(x).
\]
The (weighted) Vandermonde variety \( V(a) \) contains \( x \) and encodes the structure of connected components.

By theorem \ref{thm:arnold}, the minimum of \( p_{d+1}^{(m)} \) must be attained on a face corresponding to some \( \textsc{Res}(\lambda) \), where \( \lambda \in \CompMax(n,d) \). This ensures that every connected component of \( S_c \cap W_{(i,i+1)} \) contains a path leading to \( S_c^{\textsc{Res}(\lambda)} \), completing the proof.
\end{proof}
We will use Lemma \ref{lem:arnold2} to decide if a given point $x\in \Ww_c$ is connected to a wall $W_c^{(i,i+1)}$. Indeed, the following algorithm is a clear consequence.
\begin{algorithm}
\caption{{\sf Connected\_Wall}(\( f, x, i \))}
\label{algorithm:connected_wall}
\begin{itemize}
    \item[{\bf Input:}] \( \f = (f_1,\dots,f_s) \in \mathbb{R}[\mathbf{X}]^{S_n} \) is a set of symmetric polynomials of degree at most \( d \), defining a semi-algebraic set \( S \).  
    A point \( x \in W_c \) and \( i \in \mathbb{N} \).
    \item[{\bf Output:}] {\bf true} if \( x \) can be connected in \( S_c \) to the wall \( W_c^{(i,i+1)} \), otherwise return {\bf false}.
\end{itemize}

\vspace{-0.1cm}
\noindent\rule{16.5cm}{0.5pt}

\begin{enumerate}
    \item compute a list \( L \) of alternate odd compositions of \( n \) into \( d \) parts.
    \item {\bf for} each \( \lambda \) in \( L \), compute the modified polynomial set:
    \[
    \f^{[\textsc{Res}_i(\lambda)]} := (f_1^{[\textsc{Res}_i(\lambda)]}, \dots, f_s^{[\textsc{Res}_i(\lambda)]}).
    \]
    \item use {\sf Normal\_Connect\_Points}\( (\f^{[\textsc{Res}_i(\lambda)]}) \) to compute a list 
    \[
    L^{[\textsc{Res}_i(\lambda)]}
    \]
    of points in \( S_c^{[\textsc{Res}_i(\lambda)]} \).
    \item {\bf for} each \( z \in L^{[\textsc{Res}_i(\lambda)]} \):
    \begin{enumerate}
        \item {\bf if} {\sf Connectivity\_Symmetric\_Canonical}\( (\f^{[\textsc{Res}_i(\lambda)]}, (x, z)) \) is {\bf true}, return {\bf true}.
    \end{enumerate}
    \item return {\bf false}.
\end{enumerate}

\end{algorithm}
\begin{lemma}\label{lem:wall}
The {\sf Connected\_Wall} algorithm correctly decides if the connected component of $x\in S_c$ is path-connected to the wall $\Ww^{(i,i+1)}$. The worst-case complexity is bounded by 
\[
O\left( \binom{n - \lceil d/2 \rceil - 1}{\lfloor d/2 \rfloor - 1} \cdot d^{4d+1} \log(d) \cdot \left( (s d)^n n^{d^2} + s^{d+1} d^{O(d^2)} \right) \right).
\]
\end{lemma}
\begin{proof}
It follows directly from Lemma \ref{lem:arnold2} that we can restrict  to the different $d$-dimensional faces $\Ww_c^{[\textsc{Res}_i(\lambda)]}$. If one of these contains a point $z$ connected in $S_c$ to $x$ we are and we have established a connection to $\Ww^{(i,i+1)}$. However, in the case when non of the connected components in one $\Ww_c^{[\textsc{Res}_i(\lambda)]}$ can be connected in $S_c$ to $x$ we have established that no path connection exists to $\Ww_c^{(i,i+1)}$. Since every one of the calls of {\sf Normal\_Connect\_Points} has as an input a $d$-dimensional set, we get easily deduce that the driving factor comes from the number of elements in $\CompMax$ and get the announced bound.
\end{proof}
Finally, this gives rise to the overall algorithm.

\begin{algorithm}
\caption{{\sf Connectivity\_Symmetric}($\f, (x, y)$)}
\label{algorithm:general}
\begin{itemize}
    \item[{\bf Input:}] \( \f = (f_1,\dots,f_s) \in \mathbb{R}[\mathbf{X}]^{S_n} \) is a set of symmetric polynomials of degree at most \( d \), defining a semi-algebraic set \( S \).  
    Two points \( x, y \in \mathbb{R}^n \) are given, with \( x \in W_c \) (the Weyl chamber).
    \item[{\bf Output:}] {\bf true} if \( x \) and \( y \) are in the same connected component of \( S \), otherwise return {\bf false}.
\end{itemize}

\vspace{-0.1cm}
\noindent\rule{16.5cm}{0.5pt}
\begin{enumerate}
    \item compute a vector \( y' \) and a sequence of adjacent transpositions \( T \) by calling the algorithm {\sf Minimal\_Adjacent\_Transpositions}(\( y \)).
    \item {\bf if} {\sf Connectivity\_Symmetric\_Canonical}(\( \f, (x, y') \)) is {\bf false}, then return {\bf false}.
    \item compute a list \( L \) of alternate odd compositions of \( n \) into \( d \) parts.
    \item {\bf for} each \( (i,i+1) \in T \) {\bf do}:
    \begin{enumerate}
        \item {\bf if} {\sf Connected\_Wall}(\( x, i \)) is {\bf false}, return {\bf false}.
    \end{enumerate}
    \item return {\bf true}.
\end{enumerate}

\end{algorithm}
\begin{theorem}
Algorithm \ref{algorithm:general} correctly decides if two points $x$ and $y$ are connected within a symmetric semi-algebraic set $S$. If $s$ is the number of symmetric polynomials in the description of $S$ and if their maximal degree is fixed to be $d$, the worst-case complexity of the {\sf Connectivity\_Symmetric} algorithm is given by:

\[
O\left( s^d n^{d^2} d^{O(d^2)} \right).
\]

\end{theorem}
\begin{proof}
The correctness of the algorithm builds on Theorem \ref{thm:connected_mirrored_space}. If {\sf Connectivity\_Symmetric\_Canonical}(\( \f, (x, y') \)) is {\bf false}, the two points are not connected. If {\sf Connectivity\_Symmetric\_Canonical}(\( \f, (x, y') \)) is {\bf true} Theorem \ref{thm:connectivity} tells us that we need to check if the unique $\sigma\in S_n$ that connects maps $y$ to $y'$ can be generated by the walls to which $x$ connects in $S_c$. Following Lemma \ref{lem:wall} we know that Algorithm {\sf Connected\_Wall} is able to decide this for each of the relevant walls. This guarantees that Algorithm \ref{algorithm:general} will return true if and only if they are connected. Finally, we obtain the overall complexity directly as $$O\left( n^2 + n^{d^2} + n^2 \cdot \binom{n - \lceil d/2 \rceil - 1}{\lfloor d/2 \rfloor - 1} \cdot d^{4d+1} \log(d) \cdot \left( (s d)^d n^{d^2} + s^{d+1} d^{O(d^2)} \right) \right)$$ which yields the bound announced.
\end{proof}
\section{Conclusion and outlook}
In this article, we have presented the first polynomial-time algorithm for deciding connectivity in a symmetric semi-algebraic set. The key geometric insight underlying our approach is a dimension reduction technique, which is feasible due to the low degree of the defining polynomials. Notably, this degree restriction can be bypassed if a priori knowledge of the decomposition of symmetric polynomials into power sum polynomials is available. Moreover, alternative families of symmetric polynomial generators can also be useful. For instance, as demonstrated in \cite{riener2024linear}, a special representation in elementary symmetric polynomials leads to analogous results. The methods developed here can be generalized to this setting with only minor modifications. However, at present, we are still constrained by the specific choice of representation. An important open question is whether a sub-exponential algorithm can be obtained even in cases where no "sparse" representation in power sums or elementary symmetric polynomials is available. In such situations, the recent algorithms proposed by \cite{faugere2020computing} provide a promising direction. Furthermore, we are limited to symmetric semi-algebraic described by symmetric polynomials. While every symmetric semi-algebraic set can be described using symmetric polynomials, transforming a general description into a symmetric one may lead to an increase in polynomial degree. For a particular class of symmetric semi-algebraic sets not originally defined by symmetric polynomials, Safey El Din and the first author \cite{riener2018real} established a result that could also be utilized for dimension reduction. It would be highly interesting to explore, in a broader context, the extent to which a polynomial-time algorithm for connectivity—and more generally, for topological queries—remains feasible in the setting of symmetric sets that are not explicitly given by symmetric polynomials.  Finally, as seen in the section on Mirrored Spaces, much of our theory extends naturally to Coxeter groups. In fact, Friedl, Sanyal, and the first author have shown in \cite{friedl2018reflection} that a weaker analog of Arnold’s theorem holds for all finite reflection groups. Establishing a general framework for our polynomial-time algorithm in the context of all Coxeter groups remains an intriguing avenue for future research. 

\section*{Acknowledgements}
The authors thank Saugata Basu and Nicolai Vorobjov for their valuable comments on an earlier version of this work, which contributed to both the improvement of the results and the clarity of the presentation.
\bibliographystyle{plain} \bibliography{biblio}
\end{document}